\documentclass[11pt]{llncs}
\usepackage{amssymb,amsmath,enumerate,graphicx,hyperref,amsfonts,latexsym,url}
\usepackage[utf8]{inputenc}
\usepackage[T1]{fontenc}        
\usepackage[normalem]{ulem}
\usepackage[margin=2cm,font=small,labelfont=bf]{caption}

\advance\hoffset by -0.5in
\advance\textwidth by 1in
\advance\voffset by -0.5in
\advance\textheight by 1in

\usepackage{tikz}
\usetikzlibrary{snakes,shapes}
\tikzstyle{tre}=[circle,draw,minimum size=2.5mm,inner sep=0.2pt]%fill=green!50,

\renewcommand{\leq}{\leqslant}
\renewcommand{\geq}{\geqslant}
\newcommand{\BT}{\mathcal{T}}
\newcommand{\LBT}{\mathcal{OT}}
\newcommand{\child}{\mathrm{child}}

\newcommand{\RR}{\mathbb{R}}
\newcommand{\ZZ}{\mathbb{Z}}

\begin{document}
\title{The  probabilities of trees and cladograms under Ford's $\alpha$-model}

\author{Tom\'as M. Coronado\and Arnau Mir\and Francesc Rossell\'o}
\institute{Balearic Islands Health Research Institute (IdISBa) and Department of Mathematics and
  Computer Science,
  University of the Balearic Islands,
  E-07122 Palma, Spain, 
\texttt{\{tomas.martinez,arnau.mir,cesc.rossello\}@uib.es}}
\maketitle

\begin{abstract}
We give correct explicit formulas for the probabilities of  rooted binary  trees and cladograms under Ford's $\alpha$-model. 
 \end{abstract}

\section{Introduction}
\noindent  The study of random growth models of rooted phylogenetic trees and the statistical properties of the shapes of the phylogenetic trees they produce was initiated almost one century ago by Yule \cite{Yule} and it has gained momentum in the last 20 years: see, for instance,  \cite{Ald1,BF,Ford1,Egui,Kirk,Po16,SV}.  The final goal of this line of research is to understand the relationship between the forces that drive evolution and the topological properties of ``real-life'' phylogenetic trees  \cite{BF,MH}; see also \cite[Chap. 33]{fel:04}. One of the most popular such models is Ford's $\alpha$-model for rooted binary  trees \cite{Ford1}, a parametric model that generalizes both the uniform model (where new leaves are added equiprobably to any arc, giving rise to the uniform probability distribution on the sets of binary rooted phylogenetic trees, or \emph{cladograms}, with a fixed set of taxa) and Yule's model (where new leaves are added equiprobably only to \emph{pendant} arcs, that is, to arcs ending in leaves) by allowing a different probability for the addition of new leaves to pendant or to {internal} (i.e., non-pendant) arcs.

When models like Ford's are used to contrast topological properties of phylogenetic trees contained in databases like  TreeBase (\url{https://treebase.org}), only
their general properties (moments, asymptotic behavior) are employed. But, in the course of a research where we have needed to compute the probabilities of several specific cladograms under this model,  we have noticed that the explicit formulas that Ford gives  in \cite[\S 3.5]{Ford1} for the probabilities of cladograms and (unlabeled rooted binary) trees (see Props.\ 29 and 32 therein, with the definition of $\widehat{q}$ given in page 30)  are wrong, failing for some trees with $n\geq 8$ leaves. 

So, to help the future user of Ford's model, in this paper we give the correct explicit formulas for these probabilities. This paper is accompanied by the GitHub page \url{https://github.com/biocom-uib/prob-alpha} where the interested reader can find a SageMath \cite{sage} module to compute these probabilities, and their 
explicit values on the sets $\BT_n$ of cladograms with $n$ leaves labeled $1,\ldots,n$, for every $n$ from 2 to 8. 

\section{Preliminaries}
\subsection{Definitions, notations, and conventions}
\noindent Throughout this paper, by a \emph{tree} $T$ we mean a rooted binary tree. We shall denote its root and  its set of  \emph{internal nodes} (i.e., those nodes that are not leaves) by  $r_T$  and $V_{int}(T)$, respectively. 
We understand $T$ as a directed graph, with its arcs pointing away from the root. The \emph{children} of an internal node $u$ 
are those nodes $v$ such that $(u,v)$ is an arc in $T$, and they form 
the set $\child(u)$. A node $x$ is a  \emph{descendant} of a node $v$ when there exists a directed path from $v$ to $x$ in $T$.  For every node $v$, the \emph{subtree $T_v$ of $T$ rooted at $v$} is the subgraph of $T$ induced on the set of descendants of $v$. 

A tree $T$ is \emph{ordered} when it is endowed with an \emph{ordering}  $\prec_v$ on every $\child(v)$.
An \emph{isomorphism} of ordered trees is an isomorphism of rooted trees that moreover preserves  these orderings. 
A  \emph{cladogram} (respectively, an \emph{ordered cladogram}) on a set $\Sigma$ is a tree (resp., an {ordered tree}) with its leaves bijectively labeled in $\Sigma$.  An  \emph{isomorphism} of cladograms (resp., of ordered cladograms)  is an isomorphism of trees (resp., of ordered trees)  that preserves the leaves' labels. We shall always identify a  tree, an ordered  tree, a cladogram, or an ordered cladogram, with its isomorphism class, and in particular we shall make henceforth  the abuse of language of saying that two of these objects, $T,T'$, \emph{are the same}, in symbols $T=T'$, when they  are (only) isomorphic. 
We shall denote trees and cladograms by means of their Newick format  \cite{Newick}. In this representation of unlabeled trees, we shall denote the leaves with a symbol $*$.

Let $\BT^*_n$, $\LBT^*_n$,  $\BT_n$, and  $\LBT_n$, respectively, stand for the sets of  trees with $n$ leaves, of ordered  trees with $n$ leaves, of cladograms on $[n]=\{1,2,\ldots,n\}$, and of ordered cladograms on $[n]$. 
There exist natural isomorphism-preserving forgetful mappings
\begin{center}
\begin{tikzpicture}[thick,>=stealth,scale=0.25]
\draw(0,0) node (lbt) {$\LBT_n$};   
\draw(4,-4) node (bt) {$\BT_n$}; 
\draw(-4,-4) node (lbtp) {$\LBT^*_n$}; 
\draw(0,-8) node (btp) {$\BT^*_n$};  
\draw[->]  (lbt)--(bt);
\draw(3,-1.8) node  {$\pi_o$};   
\draw(-3,-1.8) node  {$\pi_*$};   
\draw(3,-6.2) node  {$\pi$};   
\draw(-3,-6.2) node  {$\pi_{o,*}$};   
\draw[->]  (lbt)--(lbtp);
\draw[->]  (bt)--(btp);
\draw[->]  (lbtp)--(btp);
\end{tikzpicture}
\end{center}
that ``forget'' the orderings or the labels of the trees. In particular, we shall say that two cladograms \emph{have the same shape} when they have the same image under $\pi$.

Let us introduce some more notations. For every node $v$ in a tree $T$, $\kappa_T(v)$ is its number of descendant leaves.
For every internal node  $v$ in an ordered tree $T$, with children $v_1\prec_v v_2$, its \emph{numerical split} is 
the ordered pair
$NS_T(v)=(\kappa_T(v_1), \kappa_T(v_2))$. If, instead, $T$ is unordered and if $\child(v)=\{v_1,v_2\}$ with 
$\kappa_T(v_1)\leq \kappa_T(v_2)$, then  $NS_T(v)=(\kappa_T(v_1), \kappa_T(v_2))$.
In both cases, the \emph{multiset of  numerical splits} of $T$ is
$NS(T)=\{NS_T(v)\mid v\in  V_{int}(T)\}$.
A \emph{symmetric branch point} in a  tree $T$ is an internal node $v$ such that, if $v_1$ and $v_2$ are its children, then 
the subtrees $T_{v_1}$ and $T_{v_2}$ of $T$ rooted at them are isomorphic.

Given two cladograms $T$ and $T'$ on $\Sigma$ and $\Sigma'$, respectively, with $\Sigma\cap\Sigma'=\emptyset$, their \emph{root join} is   the cladogram $T\star T'$ on $\Sigma\cup\Sigma'$ obtained by connecting the roots of $T$ and $T'$ to a (new) common root $r$.
 If $T,T'$ are ordered cladograms, $T\star T'$ is ordered by inheriting the orderings on $T$ and $T'$ and ordering the children of the new root $r$ as $r_T\prec_r r_{T'}$. If $T$ and $T'$ are unlabeled trees, a similar construction yields a tree $T\star T'$;  if they are moreover ordered, then $T\star T'$ becomes an ordered tree as explained above.

\subsection{The $\alpha$-model}

\noindent Ford's $\alpha$-model \cite{Ford1} defines, for every $n\geq 1$, a family of probability density functions $P^{(*)}_{\alpha,n}$ on $\BT^*_n$ that depend on one parameter $\alpha\in [0,1]$, and then it translates this family into three other families of probability density functions
$P_{\alpha,n}$ on $\BT_n$,  $P^{(o,*)}_{\alpha,n}$ on $\LBT^*_n$, and $P^{(o)}_{\alpha,n}$ on $\LBT_n$, by imposing that the probability of a tree is equally distributed among its preimages under $\pi$, $\pi_{o,*}$, or $\pi\circ\pi_o=\pi_{o,*}\circ \pi_*$, respectively.

It is well known  \cite{CS} that every $T\in \BT_n$ can be obtained  in a unique way by adding recurrently to a single node labeled 1, new leaves labeled $2,\ldots,n$ to arcs (i.e., splitting an arc $(u,v)$ into two arcs $(u,w)$ and $(w,v)$ and then adding a new arc from the inserted node $w$ to a new leaf) or to a new root  (i.e., adding a new root $w$ and new arcs from $w$ to the old root and to a new leaf).
The value of $P^{(*)}_{\alpha,n}(T^*)$ for $T^*\in \BT_n^*$ is  determined through all possible ways of constructing dendrograms with shape $T^*$ in this way. More specifically:
\begin{enumerate}
\item Start with the tree $T_1\in \BT_1$ consisting of a single node labeled 1. Let $P'_{\alpha,1}(T_1)=1$.

\item For every $m=2,\ldots,n$, let $T_{m}\in \BT_{m}$ be obtained by adding a new leaf labeled $m$ to $T_{m-1}$. Then:
$$
\hspace*{-1ex}P'_{\alpha,m}(T_{m})=\left\{\hspace*{-1ex}\begin{array}{ll}
 \dfrac{\alpha}{m-1-\alpha}\cdot P'_{\alpha,m-1}(T_{m-1})  & \mbox{if the new leaf is added to}\\[-1ex] &\mbox{an internal arc or to a new root}\\[2ex]
  \dfrac{1-\alpha}{m-1-\alpha}\cdot P'_{\alpha,m-1}(T_{m-1}) & \mbox{if the new leaf is added to}\\[-1ex] &\mbox{a pendant arc}
\end{array}\right.
$$

\item When the desired number $n$ of leaves is reached, the probability of every tree $T_n^*\in \BT_n^*$ is defined as
$$
P^{(*)}_{\alpha,n}(T_n^*)=\sum_{\pi(T_n)=T_n^*} P'_{\alpha,n}(T_n).
$$
\end{enumerate}
Once $P^{(*)}_{\alpha,n}$ is defined on $\BT^*_n$, it is transported to $\BT_n$,  $\LBT_n^*$, and $\LBT_n$  by defining the probability of an object in one of these sets as the probability of its image in $\BT^*_n$ divided by the number of preimages of this image:
\begin{itemize}
\item For every $T\in \BT_n$, if  $\pi(T)=T^*\in \BT^*_n$ and it has  $k$ symmetric branch points,  then
\begin{equation}
P_{\alpha,n}(T)=\frac{2^k}{n!}\cdot  P^{(*)}_{\alpha,n}(T^*),
\label{s->d}
\end{equation}
because  $|\pi^{-1}(T^*)|=n!/2^k$  (see, for instance, \cite[Lem. 31]{Ford1}).

\item  For every $T_o\in \LBT_n$, if $\pi_o(T_o)=T\in \BT_n$, then 
\begin{equation}
P^{(o)}_{\alpha,n}(T_o)=\frac{1}{2^{n-1}}\cdot  P_{\alpha,n}(T),
\label{od->d}
\end{equation}
because $|\pi_o^{-1}(T)|=2^{n-1}$ ($T$ has  $2^{n-1}$ different preimages under $\pi_o$, obtained by taking all possible different combinations of orderings on the $n-1$ sets $\child(v)$, $v\in V_{int}(T^*)$).

\item  For every $T_o^*\in \LBT^*_n$, if $\pi_{o,*}(T^*_o)=T^*\in \BT^*_n$ and it has $k$ symmetric branch points, then
\begin{equation}
P^{(o,*)}_{\alpha,n}(T^*_o)=\frac{1}{2^{n-k-1}}\cdot   P^{(*)}_{\alpha,n}(T^*),
\label{s->os}
\end{equation}
because $|\pi_{o,*}^{-1}(T^*)|=2^{n-1-k}$ (from the  $2^{n-1}$ possible preimages of $T^*$  under $\pi_{o,*}$, defined by all possible different combinations of orderings on the $n-1$ sets $\child(v)$, $v\in V_{int}(T^*)$, those differing only on the orderings on the children of the $k$ symmetric branch points are actually the same ordered tree).
\end{itemize}

The family $(P^{(o,*)}_{\alpha,n})_{n}$ satisfies the following useful Markov branching recurrence (in the sense of \cite[\S 4]{Ald1}).

\begin{proposition}\label{prop:1}
Let $\Gamma_\alpha:\ZZ^+\to \RR$ be the mapping defined by $\Gamma_\alpha(1)=1$ and, for every $n\geq 2$,
$\Gamma_{\alpha}(n)=(n-1-\alpha)\cdot\Gamma_{\alpha}(n-1)$. 
For every $a,b\in\ZZ^+$, let 
$$
q_{\alpha}(a,b)=\frac{\Gamma_\alpha(a)\Gamma_\alpha(b)}{\Gamma_\alpha(a+b)}\cdot \varphi_{\alpha}(a,b),
$$
where
$$
\varphi_{\alpha}(a,b)=\frac{\alpha}{2}\binom{a+b}{a}+(1-2\alpha)\binom{a+b-2}{a-1}.
$$
Then, for every $0<m<n$ and for every $T_m^* \in \LBT_m^*$ and  $T_{n-m}^*\in \LBT_{n-m}^*$,
$$
P^{(o,*)}_{\alpha,n}(T_n^*\star T_{n-m}^*)=q_{\alpha}(m,n-m)P^{(o,*)}_{\alpha,m}(T_m^*)P^{(o,*)}_{\alpha,n-m}(T_{n-m}^*).
$$
\end{proposition}

This recurrence, together with the fact that $P^{(o,*)}_{\alpha,1}$ of a single node is $1$,  implies that, for every $T_o^*\in \LBT_n^*$,  
\begin{equation}
P^{(o,*)}_{\alpha,n}(T_o^*)=\prod_{(a,b)\in NS(T_o^*)} q_{\alpha}(a,b).
\label{explLBT}
\end{equation}
For proofs of Proposition \ref{prop:1} and equation (\ref{explLBT}), see Lemma 27 and Proposition 28 in \cite{Ford1}, respectively.

\section{Main results}

\noindent Our first result is an explicit formula for $P_{\alpha,n}(T)$, for every $n\geq 1$ and $T\in \BT_n$:
\begin{proposition}\label{prop:main}
For every $T\in \BT_n$, its probability under the $\alpha$-model is
$$
P_{\alpha,n}(T)=\frac{2^{n-1}}{n!\cdot \Gamma_{\alpha}(n)}\prod_{(a,b)\in NS(T)} \varphi_{\alpha}(a,b).
$$
\end{proposition}

\begin{proof}
Given $T\in \BT_n$, let $T_o$ be any ordered cladogram such that $\pi_o(T_o)=T$, and let $T_o^*=\pi_{*}(T_o)\in \LBT_n^*$ and $T^*=\pi(T)=\pi_{o,*}(T_o^*)$. If $T^*$ has $k$ symmetric branch points, then, by equations (\ref{s->d}), (\ref{s->os}) and (\ref{explLBT}), 
$$
P_{\alpha,n}(T)=\frac{2^k}{n!}\cdot P^{(*)}_{\alpha,n}(T^*)=\frac{2^k}{n!}\cdot 2^{n-k-1}\cdot P^{(o,*)}_{\alpha,n}(T_o^*)
=\frac{2^{n-1}}{n!}  \prod_{(a,b)\in NS(T_o^*)} q_{\alpha}(a,b).
$$
Now, on the one hand, it is easy to check that 
$NS(T)=\big\{(\min\{a,b\},\max\{a,b\})\mid (a,b)\in NS(T_0^*)\big\}$,
and therefore, since $q_\alpha$ is symmetric,
$$
P_{\alpha,n}(T)=\frac{2^{n-1}}{n!}  \prod_{(a,b)\in NS(T)} q_{\alpha}(a,b).
$$
It remains to simplify this product. If, for every $v\in V_{int}(T)$, we denote its children by $v_1$ and $v_2$, then
$$
\prod_{(a,b)\in NS(T)} q_{\alpha}(a,b)=
\prod_{v\in V_{int}(T)}\frac{\Gamma_{\alpha}(\kappa_T(v_1))\Gamma_{\alpha}(\kappa_T(v_2))}{\Gamma_{\alpha}(\kappa_T(v))} \varphi_{\alpha}(NS(v)).
$$
For every $v\in V_{int}(T)\setminus\{r_T\}$, the term $\Gamma_{\alpha}(\kappa_T(v))$ appears twice in this product: in the denominator of the factor corresponding to $v$ itself and in the numerator of the factor corresponding to its parent. Therefore, all terms $\Gamma_\alpha(\kappa_T(v))$ in this product cancel except $\Gamma_{\alpha}(\kappa_T(r_T))=\Gamma_\alpha(n)$ (that appears in the denominator of its factor) and  every $\Gamma_\alpha(\kappa_T(v))=\Gamma_\alpha(1)=1$ with $v$ a leaf. Thus,
$$
P_{\alpha,n}(T)=\frac{2^{n-1}}{n!}\cdot \frac{1}{\Gamma_{\alpha}(n)} \cdot
\prod_{v\in V_{int}(T)}\varphi_{\alpha}(NS(v))
$$
as we claimed.\qed
\end{proof}

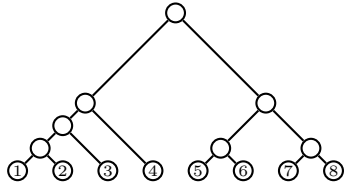
\begin{figure}[htb]
\begin{center}
\begin{tikzpicture}[thick,>=stealth,scale=0.3]
\draw(-2,0) node [tre] (0) {};   \draw (0) node {\tiny $1$};
\draw(0,0) node [tre] (1) {};     \draw (1) node {\tiny $2$};
\draw(2,0) node [tre] (2) {};    \draw (2) node {\tiny $3$};
\draw(4,0) node [tre] (3) {};    \draw (3) node {\tiny $4$};
\draw(6,0) node [tre] (4) {};    \draw (4) node {\tiny $5$};
\draw(8,0) node [tre] (5) {};    \draw (5) node {\tiny $6$};
\draw(10,0) node [tre] (6) {};    \draw (6) node {\tiny $7$};
\draw(12,0) node [tre] (7) {};    \draw (7) node {\tiny $8$};
\draw(11,1) node[tre] (a) {};
\draw(7,1) node[tre] (b) {};
\draw(9,3) node[tre] (c) {};
\draw(1,3) node[tre] (d) {};
\draw(0,2) node[tre] (e) {};
\draw(-1,1) node[tre] (f) {};
\draw(5,7) node[tre] (r) {};
\draw  (r)--(d);
\draw  (d)--(e);
\draw  (d)--(3);
\draw  (e)--(f);
\draw  (e)--(2);
\draw  (f)--(0);
\draw  (f)--(1);
\draw  (r)--(c);
\draw  (c)--(a);
\draw  (c)--(b);
\draw  (b)--(4);
\draw  (b)--(5);
\draw  (a)--(6);
\draw  (a)--(7);
\end{tikzpicture}
\caption{The cladogram with Newick string $((1,(2,(3,4))),((5,6),(7,8)));$. }
\label{tree1}
\end{center}
\end{figure}

\begin{remark}
Ford states (see \cite[Prop. 32 and page 30]{Ford1}) that if $T\in \BT_n$, then
$$
P_{\alpha,n}(T)=\frac{2^k}{n!}\prod_{(a,b)\in NS(T)} \widehat{q}_{\alpha}(a,b) 
$$
where $k$ is the number of symmetric branching points in $T$ and
$$
\widehat{q}_{\alpha}(a,b)=\left\{\begin{array}{ll}
2q_{\alpha}(a,b) & \mbox{if $a\neq b$}\\
q_{\alpha}(a,b) & \mbox{if $a=b$}
\end{array}\right. 
$$
If we simplify $\prod_{(a,b)\in NS(T)} \widehat{q}_{\alpha}(a,b)$ as in the proof of Proposition \ref{prop:main}, this formula for $P_{\alpha,n}(T)$ becomes
\begin{equation}
P_{\alpha,n}(T)=\frac{2^{k+m}}{n!  \cdot\Gamma_{\alpha}(n)}  \cdot\prod_{(a,b)\in NS(T)} \varphi_{\alpha}(a,b)
\label{eq:ford}
\end{equation}
where
$m$ is the number of internal nodes whose children have different numbers of descendant leaves.
This formula does not agree with the one given in Proposition \ref{prop:main} above, because 
$$
\begin{array}{rl}
k+m=n-1-\big|\{v\in V_{int}(T)\mid &\  \child(v)=\{v_1,v_2\}\mbox{ and } \kappa_T(v_1)= \kappa_T(v_2)\\
&   \mbox{ but }
\pi(T_{v_1})\neq \pi(T_{v_2})\}\big|
\end{array}
$$
and, hence, it may happen that $k+m<n-1$. The first example of a cladogram with this property  (and the only one, up to relabelings, with at most 8 leaves)  is the tree $\widetilde{T}\in \BT_8$ with Newick string $((1,(2,(3,4))),((5,6),(7,8)));$ depicted in Fig. \ref{tree1}. For this tree, our formula gives (see (8.22) in the document \url{https://github.com/biocom-uib/prob-alpha/blob/master/ProblsAlpha.pdf})
$$
P_{\alpha,8}(\widetilde{T})=\frac{(1-\alpha)^2(2-\alpha)}{126(7-\alpha)(6-\alpha)(5-\alpha)(3-\alpha)}
$$
while expression  (\ref{eq:ford}) assigns to $\widetilde{T}$ a probability of half this value:
\begin{equation}
\frac{(1-\alpha)^2(2-\alpha)}{252(7-\alpha)(6-\alpha)(5-\alpha)(3-\alpha)}.
\label{wrong}
\end{equation}
This last figure cannot  be right, for several reasons.
Firstly, by  \cite[\S 3.12]{Ford1}, when $\alpha=\frac{1}{2}$, Ford's model is equivalent to the uniform model, where every dendrogram in $\BT_n$ has the same probability 
has the same probability 
$$
\frac{1}{|BT_n|}=\frac{1}{(2n-3)!!}
$$
and when $\alpha=0$, Ford's model gives rise to the Yule model   \cite{Harding71,Yule}, where the probability of every $T\in\BT_n$ is
$$
P_Y(T)=\frac{2^{n-1}}{n!}\prod_{v\in V_{int}(T)}\frac{1}{\kappa_T(v)-1}.
$$
In particular, $P_{\frac{1}{2},8}(\widetilde{T})$ should be equal to $1/135135$ and $P_{0,8}(\widetilde{T})$ should be equal to $1/19845$.
Both values are consistent with our formula, while  expression (\ref{wrong}) yields half these values.

As a second reason, which can checked using a symbolic computation program, let us mention that  if we take expression (\ref{wrong}) as the probability of $\widetilde{T}$ and we assign to all other dendrograms in $\BT_8$ the probabilities computed with Proposition \ref{prop:main}, which agree on them with the values given by  equation (\ref{eq:ford}) (they are given in the aforementioned supplementary file), these probabilities do not add up 1.
\end{remark}

Now, the family of density functions $(P_{\alpha,n})_n$ satisfies the following Markov branching recurrence.

\begin{corollary}\label{cor:sm}
For every $0<m<n$ and for every $T_m\in \BT_m$ and  $T_{n-m}\in \BT_{n-m}$,
$$
P_{\alpha,n}(T_m\star T_{n-m})=\frac{2q_{\alpha}(m,n-m)}{\binom{n}{m}}P_{\alpha,m}(T_m)P_{\alpha,n-m}(T_{n-m}).
$$
\end{corollary}

\begin{proof}
If $T_m\in \BT_m$ and  $T_{n-m}\in \BT_{n-m}$, then
$$
\begin{array}{l}
\displaystyle P_{\alpha,m}(T_m)=\frac{2^{m-1}}{m!\Gamma_\alpha(m)}\prod_{(a,b)\in NS(T_m)}\varphi_{\alpha}(a,b)\\[1ex]
\displaystyle P_{\alpha,n-m}(T_{n-m})=\frac{2^{n-m-1}}{(n-m)!\Gamma_\alpha(n-m)}\prod_{(a,b)\in NS(T_{n-m})}\varphi_{\alpha}(a,b)
\end{array}
$$
and
$$
\begin{array}{l}
P_{\alpha,n}(T_m\star T_{n-m}) \displaystyle =\frac{2^{n-1}}{n!\Gamma_\alpha(n)}\prod_{(a,b)\in NS(T_m\star T_{n-m})}\varphi_{\alpha}(a,b)\\
\qquad \displaystyle =\frac{2^{n-1}}{n!\Gamma_\alpha(n)}\varphi_{\alpha}(m,n-m)\Big(\prod_{(a,b)\in NS(T_m)}\varphi_{\alpha}(a,b)\Big)\Big(\prod_{(a,b)\in NS(T_{n-m})}\varphi_{\alpha}(a,b)\Big)\\
\qquad  \displaystyle =\frac{2^{n-1}}{n!\Gamma_\alpha(n)}\varphi_{\alpha}(m,n-m)\frac{m!\Gamma_\alpha(m)}{2^{m-1}}P_{\alpha,m}(T_m)\frac{(n-m)!\Gamma_\alpha(n-m)}{2^{n-m-1}}P_{\alpha,n-m}(T_{n-m})
\\
\qquad  \displaystyle = \frac{2q_{\alpha}(m,n-m)}{\binom{n}{m}}P_{\alpha,m}(T_m)P_{\alpha,n-m}(T_{n-m})
\end{array}
$$
as we claimed. \qed\end{proof}

Combining Proposition \ref{prop:main} and equation (\ref{s->d}) we obtain the following result:

\begin{corollary}\label{cor:shape}
For every $T^*\in \BT_n^*$  with $k$ symmetric branch points, 
$$
P^{(*)}_{\alpha,n}(T^*)=\frac{2^{n-k-1}}{\Gamma_{\alpha}(n)}\prod_{(a,b)\in NS(T^*)} \varphi_{\alpha}(a,b).
$$
\end{corollary}

This formula does not agree, either, with the one given in \cite[Prop. 29]{Ford1}: the difference lies again in the same factor of 2 to the power of the number of internal nodes that are not symmetric branch points but whose children have the same number of descendant leaves.

\begin{figure}[htb]
\begin{center}
\begin{tikzpicture}[thick,>=stealth,scale=0.3]
\draw(2,0) node [tre] (2) {};  
\draw(4,0) node [tre] (3) {};  
\draw(6,0) node [tre] (4) {};  
\draw(8,0) node [tre] (5) {};  
\draw(10,0) node [tre] (6) {};  
\draw(12,0) node [tre] (7) {};  
\draw(11,1) node[tre] (a) {};
\draw(10,2) node[tre] (b) {};
\draw(5,1) node[tre] (c) {};
\draw(8,4) node[tre] (d) {};
\draw(7,5) node[tre] (e) {};
\draw  (e)--(2);
\draw  (e)--(d);
\draw  (d)--(b);
\draw  (d)--(c);
\draw  (c)--(3);
\draw  (c)--(4);
\draw  (b)--(5);
\draw  (b)--(a);
\draw  (a)--(6);
\draw  (a)--(7);
\end{tikzpicture}\qquad
\begin{tikzpicture}[thick,>=stealth,scale=0.3]
\draw(2,0) node [tre] (2) {};  
\draw(4,0) node [tre] (3) {};  
\draw(6,0) node [tre] (4) {};  
\draw(8,0) node [tre] (5) {};  
\draw(10,0) node [tre] (6) {};  
\draw(12,0) node [tre] (7) {};  
\draw(11,1) node[tre] (a) {};
\draw(10,2) node[tre] (b) {};
\draw(9,3) node[tre] (c) {};
\draw(3,1) node[tre] (d) {};
\draw(7,5) node[tre] (e) {};
\draw  (e)--(c);
\draw  (e)--(d);
\draw  (d)--(2);
\draw  (d)--(3);
\draw  (c)--(b);
\draw  (c)--(4);
\draw  (b)--(5);
\draw  (b)--(a);
\draw  (a)--(6);
\draw  (a)--(7);
\end{tikzpicture}
\caption{The trees with Newick strings $(*,((*,*),(*,(*,*))));$ (left) and $((*,*),(*,(*,(*,*))));$ (right). }
\label{tree2}
\end{center}
\end{figure}
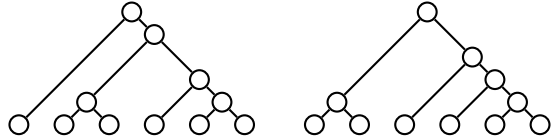

\begin{remark}
Against what is stated in \cite{Ford1}, $P^{(*)}_{\alpha,n}$ does not satisfy any Markov branching recurrence. Indeed, let $T_m,T_m'\in \BT^*_m$ be any  two different  trees with the same numbers $m$ of leaves and $k$ of symmetric branch points:
for instance, the trees in $\BT^*_6$ with Newick strings $(*,((*,*),(*,(*,*))));$ and $((*,*),(*,(*,(*,*))));$ depicted in Fig. \ref{tree2}.
Then, 
$$
\begin{array}{l}
\displaystyle P_{\alpha,m}^{(*)}(T_m)=\frac{2^{m-k-1}}{\Gamma_{\alpha}(m)}\prod_{(a,b)\in NS(T_m)}\varphi_{\alpha}(a,b)\\[2ex]
\displaystyle P_{\alpha,m}^{(*)}(T'_m)=\frac{2^{m-k-1}}{\Gamma_{\alpha}(m)}\prod_{(a,b)\in NS(T'_m)}\varphi_{\alpha}(a,b).
\end{array}
$$
In this case, $T_m\star T_m\in \BT_{2m}^*$ has $2k+1$ symmetric branch  points and therefore
$$
\begin{array}{l}
\displaystyle P_{\alpha,m}^{(*)}(T_m\star T_m) =\frac{2^{2m-2k-2}}{\Gamma_{\alpha}(2m)}\prod_{(a,b)\in NS(T_m\star T_m)} \varphi_{\alpha}(a,b)\\
\qquad\qquad \displaystyle=\frac{2^{2m-2k-2}}{\Gamma_{\alpha}(2m)}\varphi_{\alpha}(m,m)\Big(\prod_{(a,b)\in NS(T_m)} \varphi_{\alpha}(a,b)\Big)^2\\
\qquad\qquad \displaystyle= \frac{2^{2m-2k-2}}{\Gamma_{\alpha}(2m)}\varphi_{\alpha}(m,m)\Big(\frac{\Gamma_{\alpha}(m)}{2^{m-k-1}}P_{\alpha,m}^{(*)}(T_m)\Big)^2\\
\qquad\qquad \displaystyle= q_{\alpha}(m,m)P_{\alpha,m}^{(*)}(T_m)P_{\alpha,m}^{(*)}(T_m)
\end{array}
$$
while $T_m\star T'_m\in \BT_{2m}^*$ has $2k$ symmetric branch  points and therefore
$$
\begin{array}{l}
\displaystyle P_{\alpha,m}^{(*)}(T_m\star T'_m)  =\frac{2^{2m-2k-1}}{\Gamma_{\alpha}(2m)}\prod_{(a,b)\in NS(T_m\star T'_m)} \varphi_{\alpha}(a,b)\\
\qquad \displaystyle=\frac{2^{2m-2k-1}}{\Gamma_{\alpha}(2m)}\varphi_{\alpha}(m,m)\Big(\prod_{(a,b)\in NS(T_m)} \varphi_{\alpha}(a,b)\Big)\Big(\prod_{(a,b)\in NS(T'_m)} \varphi_{\alpha}(a,b)\Big)\\
\qquad \displaystyle= \frac{2^{2m-2k-1}}{\Gamma_{\alpha}(2m)}\varphi_{\alpha}(m,m)\cdot \frac{\Gamma_{\alpha}(m)}{2^{m-k-1}}P_{\alpha,m}^{(*)}(T_m)\cdot \frac{\Gamma_{\alpha}(m)}{2^{m-k-1}}P_{\alpha,m}^{(*)}(T'_m)\\
\qquad \displaystyle= 2q_{\alpha}(m,m)P_{\alpha,m}^{(*)}(T_m)P_{\alpha,m}^{(*)}(T'_m)
\end{array}
$$
and $q_{\alpha}(m,m)\neq 2q_{\alpha}(m,m)$. 
\end{remark}
 
\noindent\textbf{Acknowledgments}.
This research was supported by Spanish Ministry of Economy and Competitiveness and European Regional Development Fund project DPI2015-67082-P (MINECO/FEDER). We thank G. Cardona and G. Riera for several comments on the SageMath module companion to the paper.


\begin{thebibliography}{99}

\bibitem{Ald1}
D. Aldous. Probability distributions on cladograms. In \textsl{Random Discrete Structures} (D. Aldous and R. Pemantle, eds.),  Springer-Verlag (1996), 1--18.

\bibitem{BF} M. Blum, O. Fran\c cois. Which random processes describe the tree of life? a large-scale study of phylogenetic tree imbalance. Systematic Biology 55 (2006), 685--691.

\bibitem{CS} L. L. Cavalli-Sforza,  A. Edwards, Phylogenetic analysis. Models and estimation procedures. Am. J. Hum. Genet., 19 (1967), 233--257.

\bibitem{fel:04} J.~Felsenstein.  \textsl{Inferring Phylogenies}. Sinauer Associates Inc., 2004.

\bibitem{Ford1}
D. Ford. Probabilities on cladograms: Introduction to the alpha model. PhD Thesis (Stanford University). arXiv preprint
arXiv:math/0511246 [math.PR] (2005).

\bibitem{Harding71} E. Harding, The probabilities of rooted tree-shapes generated by random bifurcation. Adv. Appl. Prob. 3  (1971), 44--77.

\bibitem{Egui}
S. Keller-Schmidt, M. Tugrul, \textsl{et al}. An age dependent branching model for macroevolution. arXiv preprint arXiv:1012.3298  (2010).


\bibitem{Kirk}
M. Kirkpatrick, M. Slatkin. Searching for evolutionary patterns in the shape of a phylogenetic tree. Evolution 47 (1993), 1171--1181.

\bibitem{MH}
A. Mooers, S. Heard.  Inferring evolutionary process from phylogenetic tree shape.  Quarterly Review of Biology 72 (1997), 31--54.

\bibitem{Newick} The Newick tree format: \url{http://evolution.genetics.washington.edu/phylip/newicktree.html} (last visited, 10/01/2018).


\bibitem{Po16} 
L. Popovic, M. Rivas. Topology and inference for Yule trees with multiple states. J. Math. Biol., 73 (2016), 1251--1291.

\bibitem{sage} SageMath, the Sage Mathematics Software System (Version 7.6),
   The Sage Developers (2017) \url{http://www.sagemath.org}.

\bibitem{SV}
R. Sainudiin, A. Véber. A Beta-splitting model for evolutionary trees. Royal Society Open Science, 3 (2016), 160016.


\bibitem{Yule} G. U. Yule, A mathematical theory of evolution based on the conclusions of Dr. J. C. Willis. Phil. Trans.   Royal Soc. (London) Series B 213 (1924), 21--87.

\end{thebibliography}
\end{document}